\documentclass{article}
\usepackage{spconf}
\usepackage{amsmath,amssymb,amsthm}
\usepackage{algorithm}
\usepackage{algorithmic}
\usepackage{cite,multirow}
\usepackage{graphicx}

\title{Computational Complexity of Cyclotomic Fast Fourier Transforms
  over Characteristic-2 Fields}

\name{Xuebin Wu and Zhiyuan Yan}
\address{Department of ECE, Lehigh University, Bethlehem, PA 18015\\
  E-mails: \{xuw207, yan\}@lehigh.edu}

\newcommand{\f}{\mathbf{f}}

\renewcommand{\L}{\mathbf{L}}

\newtheorem{thm}{Theorem}

\newtheorem{lem}[thm]{Lemma}

\begin{document}
\maketitle
\begin{abstract}
  Cyclotomic fast Fourier transforms (CFFTs) are efficient
  implementations of discrete Fourier transforms over finite fields,
  which have widespread applications in cryptography and error control
  codes. They are of great interest because of their low
  multiplicative and overall complexities. However, their advantages
  are shown by inspection in the literature, and there is no
  asymptotic computational complexity analysis for CFFTs. Their high
  additive complexity also incurs difficulties in hardware
  implementations. In this paper, we derive the bounds for the
  multiplicative and additive complexities of CFFTs, respectively. Our
  results confirm that CFFTs have the smallest multiplicative
  complexities among all known algorithms while their additive
  complexities render them asymptotically suboptimal. However, CFFTs
  remain valuable as they have the smallest overall complexities for
  most practical lengths. Our additive complexity analysis also leads
  to a structured addition network, which not only has low complexity
  but also is suitable for hardware implementations.
\end{abstract}

\section{Introduction}
Discrete Fourier transforms (DFTs) \cite{BlahutFast} have widespread
applications in error control codes and cryptography, which in turn
are important in almost all digital communication and storage
systems. For example, the syndrome decoders of Reed-Solomon codes
\cite{BlahutECC} require DFTs over finite fields to implement the
syndrome computation and Chien search efficiently (see, e.g.,
\cite{Chen2009}). Multiplications over GF$(p^m)$ can also be
implemented efficiently by DFTs via the convolution theorem
\cite{BlahutFast} when they are formulated as multiplications of
polynomials over GF$(p)$.

Recently, very long DFTs over finite fields are needed in
practice. For example, Reed-Solomon codes over GF$(2^{12})$ with
thousands of symbols are considered for hard drive and tape storage as
well as optical communication systems to increase the data
reliability, and the syndrome decoders of such codes require DFTs of
lengths up to 4095 over GF$(2^{12})$. However, direct implementations
of DFTs have quadratic complexities with the lengths of DFTs, and the
computational complexity is prohibitive for the DFTs with thousands of
symbols. Therefore we need low-complexity algorithms and efficient
hardware implementations for DFTs over finite fields.

The cyclotomic fast Fourier transforms (CFFTs), first proposed in
\cite{Trifonov2003}, have attracted a lot of attention because of
their low multiplicative and overall complexities. Though these
advantages of CFFTs have been demonstrated for short to moderate
lengths in the literature (see, e.g., \cite{Chen2009c}), it is unclear
if they still hold for large lengths. Therefore asymptotic
computational complexity analysis is required to compare the
complexities of CFFTs with other existing DFT algorithms over finite
fields \cite{Wang1988,Cantor1991,Mateer2008}, which can help system
designers to find the optimal implementation of very long DFTs.

Another issue regarding the CFFTs is their relatively high additive
complexities, which hinder their usages. Though the additive
complexities of CFFTs can be reduced by the common expression
elimination (CSE) algorithm in \cite{Chen2009c}, the lack of addition
network structure increases the difficulty of wiring and module
reusing, and introduces other problems to the hardware
implementation. Therefore, a structured additive complexity reduction
method is appreciated for CFFTs.

In this paper, we analyze the asymptotic computational complexities of
CFFTs and derive bounds on the multiplicative and additive
complexities of CFFTs. The comparisons between our results and
existing algorithms show that CFFTs have the smallest multiplicative
complexity, but their high additive complexities render them not
asymptotically optimal. However, CFFTs are still valuable as they have
the smallest overall complexities for most DFTs with practical
lengths. Our additive complexity analysis also leads to a structured
addition network, which not only has low complexity but also is
suitable for hardware implementations.


\section{Cyclotomic Fast Fourier Transforms}
\label{sec:CFFT}
To make our paper self-contained, we first review CFFTs over GF$(2^m)$
\cite{Trifonov2003} briefly in this section.  Let $\alpha \in
\mbox{GF}(2^m)$ be an element of order $n$, where $n|2^m-1$. Consider
an $n$-dimensional vector $\f=(f_0, f_1, \cdots, f_{n-1})^T$ over
GF$(2^m)$, whose \emph{polynomial representation} is given by $f(x) =
\sum_{i=0}^{n-1}f_ix^i$. The DFT of $\f$ is $\mathbf{F}=(F_0, F_1,
\cdots, F_{n-1})^T$, where $F_j=f(\alpha^{j})$.

We partition the set of integers $\{0, 1, \cdots, n-1\}$ into $k$
cyclotomic cosets modulus $n$ with respect to two as:
\begin{align*}
  &\{0\},\, \{s_1, 2s_1, \cdots, 2^{m_1-1}s_1\},\,\{s_2, 2s_2, \cdots,
  2^{m_2-1}s_2\},\,\cdots \\
  &\{s_{k-1}, 2s_{k-1},\cdots, 2^{m_{k-1}-1}s_{k-1}\},
\end{align*}
where $m_i$ is the size of the $i$-th cyclotomic coset, and
$s_i=2^{m_i}s_i \pmod{n-1}$ is its representative. Then the polynomial
$f(x)$ can be decomposed as $f(x)=\sum_{i=0}^{k-1}L_i(x^{s_i})$, where
$ L_i(y)=\sum_{j=0}^{m_i-1}f_{2^js_i\,\mathrm{mod}\,n}\,y^{2^j}$. The
polynomial $L_i(x)$ has a property such that $L_i(x+y)=L_i(x)+L_i(y)$,
for $x, y\in \textrm{GF}(2^m)$, which is used to reduce the DFT
computational complexity in CFFT.

The element $F_j$ in the DFT result $\mathbf{F}$ can be expressed as
$F_j=f(\alpha^j)=\sum_{i=0}^{k-1}L_i(\alpha^{js_i})$. By the normal
basis theorem \cite{Lidl1996}, there is a normal basis
$\{\gamma_i^{2^0}, \gamma_i^{2^1}, \cdots, \gamma_i^{2^{m_i-1}}\}$ of
GF$(2^{m_i})$, such that $\alpha^{s_i} \in \mbox{GF}(2^m)$ can be
represented as $\sum_{s=0}^{m_i-1}a_{i,j,s}\gamma_{i}^{2^s}$, where
$a_{i,j,s}$ is binary. Therefore
\begin{align*}
  f(\alpha^j) =\sum_{i=0}^{k-1}\sum_{s=0}^{m_i-1}a_{i,j,s}\left(\sum_{t=0}^{m_i-1}\gamma_i^{2^{s+t\,\mathrm{mod}\,m_i}}f_{s_i2^t\,\mathrm{mod}\,n}\right).
\end{align*}
Writing in the matrix form, we have that the DFT can be computed by
$\mathbf{F}=\mathbf{A}\mathbf{L}\f'$, where $\f'$ is a rearrangement
of $\f$ according to the cyclotomic coset, i.e., $\f'=(\f_0', \f_1',
\cdots, \f_{k-1}')^T$ with $\f_i'=(f_{s_i}, f_{2s_i}, \cdots,
f_{2^{m_i-1}s_i})$, $\mathbf{A}$ is an $n\times n$ binary matrix
accumulating the coefficients $a_{i,j,s}$, and $\mathbf{L}$ is a block
diagonal matrix with sub-matrices $\L_i$'s on its diagonal. Block
$\L_i$ is an $m_i\times m_i$ circulant matrix corresponding to a
cyclotomic coset of size $m_i$, and it is generated from a normal
basis $\{\gamma_i^{2^0}, \gamma_i^{2^1}, \cdots,
\gamma_i^{2^{m_i-1}}\}$ of GF$(2^{m_i})$.
Therefore, the multiplication between $\L_i$ and $\f_i'$ can be
formulated as an $m_i$-point cyclic convolution between
$\mathbf{b}_i=(\gamma_i^{2^0}, \gamma_i^{2^{m_i-1}},
\gamma_i^{2^{m_i-2}}, \cdots, \gamma_i^{2^{1}})^T$ and $\f_i'$. Since
the matrix $\mathbf{A}$ is binary, the product between $\mathbf{A}$
and the vector $\mathbf{v}=\L\f'$ can be simply computed by
additions. All the multiplications needed by CFFTs are contributed by
the convolutions between $\mathbf{b}_i$ and $\f_i'$. Because the short
convolutions can be computed by efficient bilinear algorithms (see,
e.g., \cite{BlahutFast}), CFFTs have very low multiplicative
complexities. However, if implemented directly, they will have very
high additive complexities.

\section{Computational Complexities of CFFTs over Characteristic-2
  Fields}
\label{sec:complexity}
In our complexity analysis of CFFTs, we aim to theoretically show that
their multiplicative complexities are the smallest among all known
techniques and to investigate the optimality of the overall
computational complexities of CFFTs. For this effort, we focus on
CFFTs of length $n=2^m-1$ over GF$(2^m)$.

We denote the cyclotomic cosets of the set $\{0, 1, \cdots, n-1\}$
modulus $n$ with respect to two as $C_0$, $C_1$, $\cdots$, $C_{k-1}$,
and assume that $C_i$ has $m_i$ elements with a representative
$s_i$. It is required that $m_i$ divides $m$, i.e., $m_i | m$. We
divide $C_i$'s into $d$ groups --- $G_0$, $G_1$, $\cdots$, $G_{d-1}$
--- so that $C_i$'s in each group are of the same size. We denote the
size of $G_j$ as $|G_j|$.

As described in Sec.~\ref{sec:CFFT}, an $n$-point CFFT is given by
$\mathbf{AL}\mathbf{f}'$, where the matrix $\mathbf{A}$ is binary. The
product of the matrix $\mathbf{L}$ and the vector $\mathbf{f}'$, i.e.,
a vector $\mathbf{v}=\mathbf{Lf}'$, is computed via $k$ cyclic
convolutions, with $\mathbf{L}_i\mathbf{f}_i$ being an $m_i$-point
cyclic convolution. It is a well-known result that an $n$-point cyclic
convolution requires $O(n^{\log_2 3})$ multiplications and additions,
respectively \cite{Winograd1980}. The $k$ cyclic convolutions
contribute to both the multiplicative and additive complexities of the
CFFT, while computing $\mathbf{Av}$ only contributes to the additive
complexity since $\mathbf{A}$ is binary.

\subsection{Multiplicative Complexities of CFFTs over GF$(2^m)$}
\label{sec:multcfft}
By the definition of big~$O$ notation, an $m_i$-point cyclic
convolution has a multiplicative complexity less than
$cm_i^{\log_23}$, where $c$ is a constant independent with
$m_i$. Hence the total multiplicative complexity of an $n$-point CFFT
is less than $ c\sum_{i=0}^{k-1}m_i^{\log_23} $.  As introduced in the
beginning of this section, we can group the cyclotomic cosets
according to their sizes into $d$ groups, and each group $G_j$ has
$|G_j|$ cyclotomic cosets. We then have that the size of the cosets in
$G_j$, given by $g_j$, divides $m$, i.e., $g_j|m$, and also $d\le
m$. Since $\log_2 3>1$, we have
$\frac{m}{g_j}(g_j)^{\log_23}=m(g_j)^{\log_2\frac{3}{2}}\le
m(m)^{\log_2\frac{3}{2}}=m^{\log_23}$. Hence, the total multiplicative
complexity satisfies
\begin{align*}
  &c\sum_{i=0}^{k-1}m_i^{\log_23}=c\sum_{j=0}^{d-1}|G_j|g_j^{\log_23}\\
  =&c\sum_{j=0}^{d-1}\lfloor\frac{|G_j|g_j}{m}\rfloor
  \frac{m}{g_j}g_j^{\log_23}+c\sum_{j=0}^{d-1} (|G_j| \bmod m/g_j)
  g_j^{\log_23}\\
  \le&2c\frac{2^m-1}{m} m^{\log_23},
\end{align*}
when $m\ge 4$ since $d\le m \le (2^m-1)/{m}$ in such cases. Since we
are considering the asymptotic complexity, we do not need to consider
the case $m<4$. The total multiplicative complexity of an $n$-point
CFFT is thus $O(n(\log_2n)^{\log_2\frac{3}{2}})$ since $m =
\log_2(n+1)$.

Unfortunately, this bound on multiplicative complexities of CFFTs
cannot be generalized to an arbitrary $n$. This can be shown by
counterexamples. For instance, for some lengths (say $n=11$ or $13$),
the set of integers $\{0, 1, \cdots, n-1\}$ is partitioned into only
two cyclotomic cosets, $\{0\}$ and $\{1,2,\cdots,n-1\}$. Hence, the
total multiplicative complexities of CFFTs of these lengths are on the
order of $O(n^{\log_23})$.

\subsection{Additive Complexities of CFFTs over GF$(2^m)$}
\label{sec:addcfft}

Both the convolutions and multiplication between the binary matrix
$\mathbf{A}$ and the vector $\mathbf{v}=\mathbf{L}\mathbf{f}'$
contribute to the additive complexity of an $n$-point CFFT over
GF$(2^m)$ with $n=2^m-1$. Since the additive and multiplicative
complexities of a cyclic convolution have the same order, the total
additive complexity contributed by the convolutions is
$O(n(\log_2n)^{\log_2\frac{3}{2}})$.  However, the additive
complexities of CFFTs are dominated by the computing
$\mathbf{Av}$. Since $\mathbf{A}$ consists of only $0$ and $1$, only
addition is needed to compute $\mathbf{Av}$. We will derive the
additive complexity of $\mathbf{A}\mathbf{v}$.

The Four-Russian algorithm \cite{Aho1974} is an efficient algorithm
for binary matrix multiplication, and it requires $O(n^2/\log_2 n)$
additions for a multiplication between an $n\times n$ matrix and an
$n$-dimensional vector, referred to as $n\times n$ matrix vector
product (MVP). However, it does not consider the structure of
$\mathbf{M}$. Next we further reduce the additive complexity of
computing $\mathbf{A}\mathbf{v}$ by exploring the inner structure of
the matrix $\mathbf{A}$.

As shown in Sec.~\ref{sec:CFFT}, for an $n$-point CFFT over GF$(2^m)$
where $n=2^m-1$, the matrix $\mathbf{A}$ can be partitioned into
$1\times k$ blocks, and each block $\mathbf{A}_i$ is of size
$(2^{m}-1)\times m_i$, and its row $j$ is the representation of
$\alpha^{js_i}$ under a normal basis in the field GF$(2^{m_i})$, where
$\alpha$ is an element in GF$(2^m)$ of order $n$.


We first rearrange the rows of the matrix $\mathbf{A}$ according to
the cyclotomic cosets. The rearrangement will result in a new matrix
$\mathbf{A}'$, which can be partitioned into $k\times k$ blocks. Each
block $\mathbf{A}'_{ij}$ is of size $m_i\times m_j$, and row $t$
in the block $\mathbf{A}'_{ij}$ is the representation of
$\alpha^{2^ts_is_j}$ under a normal basis in GF$(2^{m_j})$. By the
property of normal bases, we know that row $t$ is just a right
cyclic shift of the previous row, and hence $\mathbf{A}'_{ij}$ is a
cyclic matrix \cite{Fedorenko2006}. We then partition the vector
$\mathbf{v}$ into $k$ blocks correspondingly, and the block
$\mathbf{v}_i$ has $m_i$ elements. The product
$\mathbf{Av}$ can be recovered by reordering the elements in the
vector $\mathbf{A}'\mathbf{v}$.

All those $m_i\times m_j$ blocks can be extended to $m \times m$
matrices while keeping the cyclic property.  Since $m_i$ and $m_j$ are
all factors of $m$, we first partition an $m\times m$ matrix into
$\frac{m}{m_i}\times \frac{m}{m_j}$ blocks of size $m_i\times m_j$,
and then set each block to $\mathbf{A}_{ij}$. The resulting $m\times
m$ matrix is still a cyclic matrix. After extending all the blocks to
$m\times m$ blocks in this way, we will get a $km\times km$ matrix
$\mathbf{A}''$.  To ensure that we can recover the multiplication
result $\mathbf{Av}$, we should also extend each sub-vector
$\mathbf{v}_i$ to a vector of length $m$ by padding zeros in the end,
resulting in a $km$-dimensional vector $\mathbf{v}''$.  The elements
in $\mathbf{A}''\mathbf{v}''$ corresponding to the extended rows are
simply discarded.

To utilize this cyclic sub-matrices structure, we construct a new
matrix $\mathbf{B}$ and a new vector $\mathbf{u}$ from $\mathbf{A}''$
and $\mathbf{v}''$, respectively according to the following rules:
\begin{equation}
  \label{eq:reorder}
  B_{i_2k+i_1,j_2k+j_1}=A''_{i_1m+i_2,j_1m+j_2},\,
u_{i_2k+i_1}=v''_{i_1m+i_2},
\end{equation}
where $0\le i_1, j_1 < k$, $0 \le i_2, j_2 < m$, $A''_{i,j}$,
$B_{i,j}$ are the elements in row $i$ and column $j$ in the matrix
$\mathbf{A}''$ and $\mathbf{B}$, respectively, and $u_i$ and $v''_i$
are the elements at position $i$ in the vector $\mathbf{u}$ and
$\mathbf{v}''$, respectively. The matrix $\mathbf{B}$ just reorders
the rows and columns of $\mathbf{A}''$, and reordering the vector
$\mathbf{v}''$ into $\mathbf{u}$ ensures that the product
$\mathbf{Av}$ can be extracted by reordering $\mathbf{Bu}$ without
additional computational complexity. Since $\mathbf{A}''$ contains
$k\times k$ blocks of cyclic matrices of size $m\times m$, the matrix
$\mathbf{B}$ is a block-cyclic matrix with $m\times m$ block matrices
of size $k\times k$.

Since the result $\mathbf{Av}$ can be extracted from $\mathbf{Bu}$
without any additional computational complexity, the computational
complexity of $\mathbf{Bu}$ serves as an upper bound of that of
$\mathbf{Av}$. Now let us analyze the computational complexity of
$\mathbf{Bu}$. The matrix $\mathbf{B}$ is an $m\times m$ block-cyclic
matrix, therefore it can be computed via $O(m^{\log_23})$
multiplications between a $k \times k$ matrix and a $k$-dimensional
vector and $O(m^{\log_23})$ additions of two $k$-dimensional vectors
\cite{Winograd1980}. Since the matrix $\mathbf{B}$ is a fixed one, all
the additions between $k\times k$ matrices can be precomputed, and it
does not contribute to the additive complexity. Applying the
Four-Russian algorithm, the multiplication between a $k\times k$
matrix and a $k$-dimensional vector requires $O(k^2/\log_2 k)$
additions. The addition between two $k$-dimensional vectors requires
$k$ additions, and hence the total computational complexity can be
written as
$$O(m^{\log_23}\frac{k^2}{\log_2k})+O(m^{\log_23}k)=O(m^{\log_23}\frac{k^2}{\log_2k}).$$
We need to find out the lower and upper bounds of $k$. Before giving
these bounds, let us prove two lemmas.
\begin{lem}
  \label{lem:cosetsize}
  In the cyclotomic cosets of $\{0,1,\cdots, 2^m-2\}$ modulus $2^m-1$
  with respect to two, there are at most $(2^{m_i}-1)/m_i$ cosets with
  size $m_i$, where $m_i|m$.
\end{lem}

\begin{proof}
  Consider the nonzero elements in the finite field GF$(2^m)$, which
  can be represented as $\alpha^j$ and $\alpha$ is a primitive element
  in GF$(2^m)$. By normal basis theorem \cite{Lidl1996}, there is at
  least one normal basis in GF$(2^m)$. Let us pick a normal basis
  $\{\gamma^{2^0}, \gamma^{2^1}, \cdots, \gamma^{2^{m-1}}\}$ in
  GF$(2^m)$. Each element in GF$(2^m)$ has an $m$-bit binary vector
  representation under this basis, i.e.,
  $\alpha^j=\sum_{i=0}^{m-1}b_i\gamma^{2^i}$, and
  $(b_{m-1}b_{m-2}\cdots b_0)$ is the vector representation of
  $\alpha^j$.

  It is easy to see that the vector representation of $\alpha^{2j}$ is
  just a left cyclic shift of that of $\alpha^j$. Therefore, if an
  integer $j$ is in the cyclotomic coset $C_i$, the vector
  representation of $\alpha^j$ repeats itself after $m_i$ shifts,
  where $m_i$ is the size of $C_i$. If $m_i < m$, then $m_i|m$, and
  the vector representation of $\alpha^j$ can be partitioned into
  $\frac{m}{m_i}$ blocks, all of which are identical and have the same
  size $m_i$, otherwise it cannot repeat itself after $m_i$ cyclic
  shifts. Therefore, there are at most $(2^{m_i}-1)/m_i$ cyclotomic
  cosets with size $m_i$.
\end{proof}

\begin{lem}
  \label{lem:kbound}
  $2^m-1 < km < 2(2^m-1)$, where $m$ is a positive integer and $k$ is
  the number of the cyclotomic cosets of $\{0,1, \cdots, 2^m-2\}$
  modulus $2^m-1$ with respect to two.
\end{lem}

\begin{proof}
  The lower bound of $km$ comes from the fact that $m$ is the maximum
  cyclotomic coset size. It suffices to prove the upper bound of $km$.

  Without loss of generality, we assume that the group $G_0$ contains
  the cosets with a size of $m$, and other groups contain the cosets
  with sizes less than $m$.  Therefore by Lemma~\ref{lem:cosetsize} we
  have
  \begin{align}
    km &= |G_0|m + \sum\nolimits_{j=1}^{d-1}|G_j|m \nonumber \\
    & \le (2^m-1) + m\sum\nolimits_{m_i=1,
      m_i|m}^{\lfloor\frac{m}{2}\rfloor} \frac{2^{m_i}-1}{m_i}
    \nonumber \\
    & \le (2^m-1) + m\sum\nolimits_{m_i=1}^{\lfloor\frac{m}{2}\rfloor}2^{m_i}
    \nonumber \\
    & \le (2^m-1)+m(2^{\frac{m+2}{2}}-1).
    \label{eq:coset1}
  \end{align}
  Consider the function $f(x)=x-2^{\frac{x-2}{2}}$. It is easy to
  check that $f'(x)=1-(0.5\ln 2) 2^{\frac{x-2}{2}}<0$ when $x\ge 8$,
  which means $f(x)$ is strictly decreasing when $x \ge 8$. We can
  also check that $f(9)<0$, and hence $f(x)\le f(9)<0$ when $x\ge
  9$. Therefore, for an integer $m \ge 9$, we have $f(m)<0$, and $m\le
  2^{\frac{m-2}{2}}$. Substituting this in \eqref{eq:coset1}, we have
  $km\le (2^m-1)+(2^m-2^{\frac{m-2}{2}})\le 2(2^m-1)$
  when $m\ge 10$. For $m\le 9$, the lemma can be verified by
  inspection.
\end{proof}

We have shown that the total computational complexity of evaluating
$\mathbf{Bu}$ is $O(m^{\log_23} k^2/\log_2k)$ additions, hence there
exists a constant $c$ independent of $m$ and $k$ such that the total
computational complexity is less than $cm^{\log_23} k^2/\log_2k$
additions. By Lemma~\ref{lem:kbound}, we have
\begin{align}
  cm^{\log_23} \frac{k^2}{\log_2k} 
  & \le cm^{\log_2 3}
  \frac{4(2^m-1)^2}{m^2(m+\log_2\frac{1-2^{-m}}{m})}.
  \label{eq:O1}
\end{align}
Consider the function $f(x)=2^{x}-2^{-x}-2x$. We can show that $f'(x)
= (2^x+2^{-x})\ln 2 -2 > 0$ when $x\ge 2$, and $f(3)>0$. Therefore,
$f(x)>f(3)>0$ when $x\ge3$, which implies $2^x-2^{-x}>2x$. Then we can
show that
\begin{align*}
  \log_2\frac{1-2^{-m}}{m} = \log_2 2^{-\frac{m}{2}} + \log_2
  \frac{(2^{\frac{m}{2}}-2^{-\frac{m}{2}})}{m}
  \ge -\frac{m}{2},
\end{align*}
when $m>6$. Since we are considering the asymptotic complexity, the
cases when $m\le 6$ do not need to be considered. Substituting this
result to \eqref{eq:O1}, we have
$$
cm^{\log_23} \frac{k^2}{\log_2k} \le cm^{\log_2 3}
\frac{8(2^m-1)^2}{m^3}=8c\frac{(2^m-1)^2}{m^{\log_2\frac{8}{3}}}.
$$

Since $n=2^m-1$, the additive complexity of $\mathbf{Bu}$ as well as
$\mathbf{Av}$ is upper bounded by $O(n^2/(\log_2
n)^{\log_2\frac{8}{3}})$ and  is lower than $O(n^2/\log_2 n)$, the
additive complexity of the multiplication between an arbitrary $n
\times n$ binary matrix and a vector.

\subsection{Discussions}
\label{sec:discussion}
\begin{table*}[htb]
  \centering
  \caption{Asymptotic complexities of $(2^m-1)$-point DFT algorithms and their respective restrictions. All logarithms are base
    two.}
  \label{tab:cmp}
  \small
  \begin{tabular}{|c|c|c|c|c|c|}
    \hline
    \multirow{2}{*}{Alg.} & \multicolumn{2}{c|}{Restriction} &
    \multicolumn{3}{c|}{Complexities} \\
    \cline{2-6}
    & Fields & Lengths ($n$) & Multiplicative & Additive & Total\\
    \hline
    Wang \cite{Wang1988} & GF$(2^{m})$, $m$ arbitrary & $2^m-1$ & $O(n(\log n)^2)$ &
    $O(n(\log n)^2)$ & $O(n(\log n)^3)$\\
    \hline
    Cantor \cite{Cantor1991} & GF$(2^m)$, $m$ arbitrary  & $2^m-1$ &
    $O(n^2)$ & $O(n^2)$ & $O(n^2\log n)$  \\
    \hline
    Gao \cite{Gao2001} & GF$(2^{m})$, $m=2^K$ & $2^{m}-1$ & $O(n\log n \log \log n)$
    &  $O(n\log n \log \log n)$ & $O(n(\log n)^2\log\log n)$\\
    \hline
    Mateer \cite{Mateer2008} & GF$(2^{m})$, $m=2^K$ & $2^{m}-1$ & $O(n\log n)$ &
    $O(n\log n \log \log n)$ & $O(n(\log n)^2)$\\
    \hline
    CFFTs & GF$(2^{m})$, $m$ arbitrary & $2^m-1$ & $O(n(\log n)^{\log_2\frac{3}{2}})$ &
    $O(n^2/(\log n)^{\log_2\frac{8}{3}})$ & $O(n^2/(\log n)^{\log_2\frac{8}{3}})$ \\
    \hline
  \end{tabular}
\end{table*}
To evaluate the tightness of our asymptotic bounds, in
Fig.~\ref{fig:cmp} we compare our bounds with the actual
multiplicative and additive complexities of CFFTs in
\cite{Chen2009c}. In Fig.~\ref{fig:cmp}, we scale our bounds so that
they match the actual complexities when $n=1023$.
\begin{figure}[htb]
  \centering
  \includegraphics[width=\columnwidth]{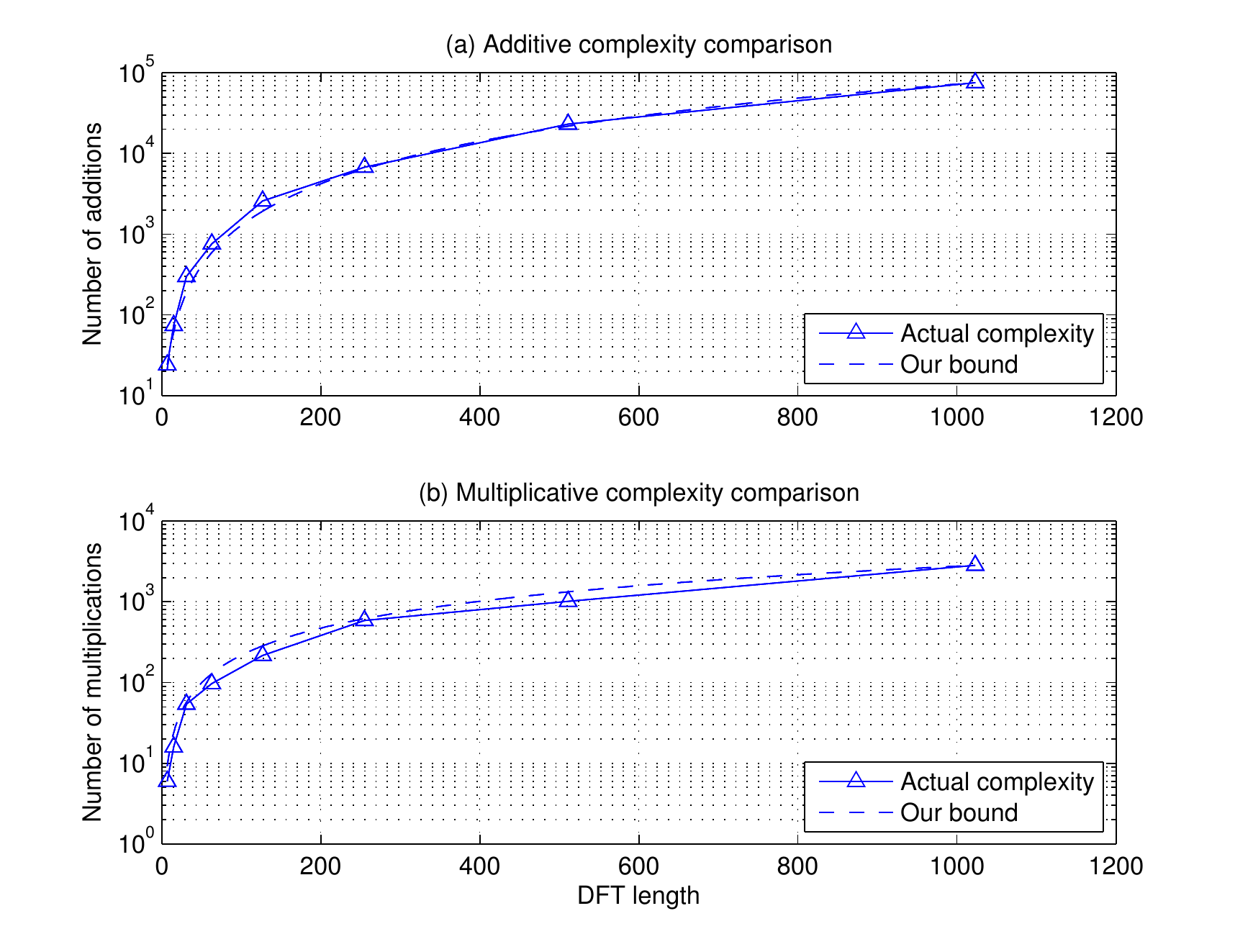}
  \caption{Comparison of the actual complexities and our bounds.}
  \label{fig:cmp}
  \vspace{-5mm}
\end{figure}
From Fig.~\ref{fig:cmp}, we can see that our bound on additive and the
multiplicative complexity is rather tight. The solid curves
corresponding to the actual complexities are very closed to the dashed
curve corresponding to the theoretical bounds. Therefore, the actual
additive and multiplicative complexities are on the order of
$O(n^2/(\log_2 n)^{\log_2\frac{8}{3}})$ and $O(n(\log_2
n)^{\log_2\frac{3}{2}})$, respectively. We remark that since we have
scaled the theoretical bounds to match the actual complexity at
certain points, it is not necessary that the computational complexity
is strictly smaller than the theoretical bound. 


We then compare asymptotic bounds on the complexities of CFFTs and
other algorithms in the literature.  In \cite{Wang1988}, a fast DFT
algorithm is proposed for GF$(2^{m})$, where $m$ can be any positive
integer. Both the additive and multiplicative complexities of this
algorithm are of $O(n(\log n)^2)$. When $m$ is a power of two, more
efficient algorithms are proposed. For example, Gao's algorithm in
\cite{Gao2001} has both the additive and multiplicative
complexities of order $O(n\log n\log\log n)$, and this result is
improved by Mateer's algorithm \cite{Mateer2008}, which reduces the
multiplicative complexity to $O(n\log n)$. When the length of the DFT
$n=s^r$ is a power of some integer $s$, \cite{Cantor1991} introduces a
fast DFT algorithm that has a computational complexity of
$rn(s-1)$. Note that this algorithm works for arbitrary algebras
rather than finite fields.

Tab.~\ref{tab:cmp} summarizes the asymptotic computational
complexities of the aforementioned algorithms when we apply them to
the DFTs with lengths of $2^m-1$ over GF$(2^m)$. To compare the total
complexities of these algorithms, we define the total complexity to be
a weighted sum of the additive and multiplicative complexities, and
assume that one multiplication over GF$(2^m)$ has the same complexity
as $2m-1$ additions over the same field. That is, the total complexity
is given by $\mbox{total}=(2m-1)\mbox{multiplicative} +
\mbox{additive}$. We note that this assumption comes from both the
hardware and software considerations \cite{Chen2009c}.  Since
we focus on $(2^m-1)$-point DFTs, we have that $m=\log_2(n+1)$ and
$2m-1$ is of order $O(\log n)$.

From Tab.~\ref{tab:cmp}, CFFTs have the lowest multiplicative
complexities among all algorithms. Furthermore, as shown in
Fig.~\ref{fig:cmp}(b), our asymptotic bound on the multiplicative
complexities of CFFTs is loose. These results confirm the advantage of
CFFTs in the multiplicative complexities. On the other hand, due to
their high additive complexities, the additive and overall
complexities of CFFTs are asymptotically suboptimal.  We emphasize the
different assumptions for the different DFT algorithms in
Tab.~\ref{tab:cmp}. For all DFT algorithms, it assumed that the length
$n$ and the size of the underlying field are such that a DFT is
well-defined. CFFTs and the fast DFT algorithm in \cite{Wang1988} have
no additional assumptions. In contrast, the other three algorithms in
Tab.~\ref{tab:cmp} all have additional constraints.  First, Cantor's
algorithm \cite{Cantor1991} requires $n=s^r$, which is often difficult
to satisfy. When $n=2^m-1$ (and other values of $n$), the only way to
satisfy this condition is $n=n^1$ due to Mih\v{a}ilescu's Theorem
\cite{Preda2004}. When $r=1$, Cantor's algorithm has a quadratic
additive and multiplicative complexities and does not have any computational
advantage. Furthermore, the algorithms in \cite{Mateer2008} work only
in a field GF$(2^{m})$ with $m=2^K$.

As shown in \cite{Chen2009c}, CFFTs have lower overall complexities
than all other DFT algorithms for most lengths up to thousands of
symbols over GF$(2^m)$ with $m\le 12$. The \emph{only exception} is
that for 255-point DFT over GF$(2^8)$, the overall complexity of
Mateer's algorithm is roughly 4\% smaller than a 255-point
CFFT. Although the overall complexities of CFFTs are asymptotically
suboptimal, CFFTs remain very significant since they have the smallest
overall complexities for most practical lengths.

\begin{figure}[htb]
  \centering
  \includegraphics[width=\columnwidth]{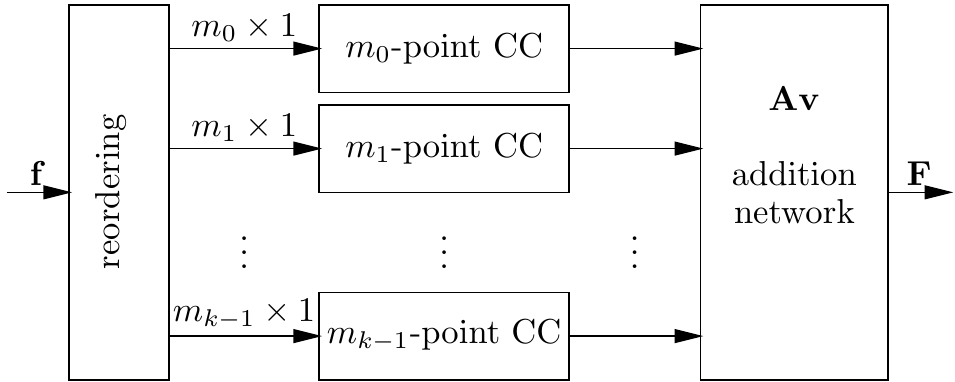}
  \caption{Implementation diagram of $(2^m-1)$-point CFFTs.}
  \label{fig:CFFT}
  \vspace{-3mm}
\end{figure}

\section{Hardware Implementations of CFFTs}
\label{sec:hardware}
The architecture of a $(2^m-1)$-point CFFT over GF$(2^m)$ is shown in
Fig.~\ref{fig:CFFT}. First, the input vector $\mathbf{f}$ is reordered
and divided into $k$ sub-vectors according to the cyclotomic
cosets. Then we perform an $m_i$-point cyclic convolution between each
sub-vector $\f_i'$ and its corresponding pre-computed vector
$\mathbf{b}_i$, as described in Sec.~\ref{sec:CFFT}. The cyclic
convolution results then go through the addition network to compute
$\mathbf{Av}$. In Fig.~\ref{fig:CFFT}, the reordering module can be
realized by wiring only and the cyclic convolution modules often can
be reused as most of the convolutions are of the same sizes. As the
computation of $\mathbf{Av}$ accounts for the majority of the total
computational complexity, the addition network in Fig.~\ref{fig:CFFT}
requires significant area and power in hardware implementations, which
makes it difficult to implement CFFTs in hardware. Though the additive
complexity of the additive network can be reduced by techniques such
as the CSE algorithm in \cite{Chen2009c}, the resulted addition
network lacks structure and hence is difficult for hardware
implementations.

\begin{figure}[htb]
  \centering
  \includegraphics[width=0.9\columnwidth]{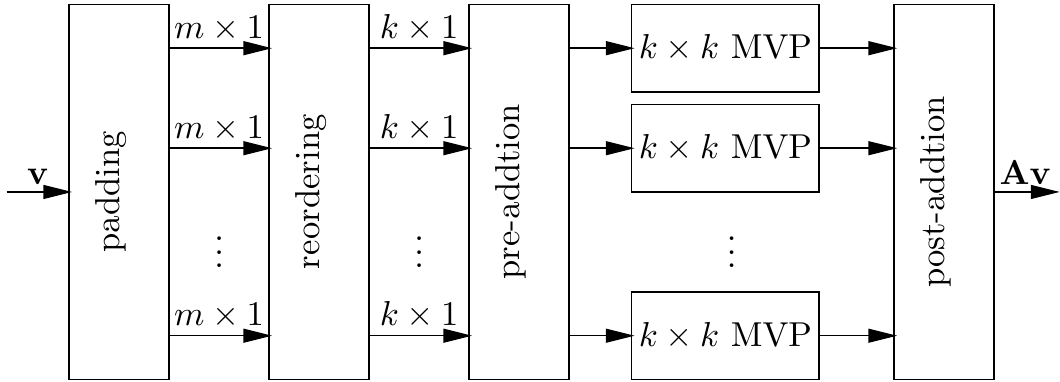}
  \caption{Our circuitry for the addition network in
    Fig.~\ref{fig:CFFT}. The reordering module has $k$ inputs and $m$
    outputs, and the pre-addition module outputs $O(m^{\log_2 3})$
    $k$-dimensional vectors. }
  \label{fig:Av}
  \vspace{-2mm}
\end{figure}

Our additive complexity analysis in Sec.~\ref{sec:addcfft} leads to a
structured addition network, which can be implemented by the
architecture shown in Fig.~\ref{fig:Av}. The vector $\mathbf{v}$, the
cyclic convolution results, is first divided into $k$ sub-vectors
according to the sizes of the cyclotomic cosets, and then each
sub-vector is extended to $m$-dimensional by padding zeros. The $k$
$m$-dimensional vectors are then reordered into $m$ $k$-dimensional
vectors according to \eqref{eq:reorder}. Since the addition network
follows the bilinear algorithm of an $m\times m$ cyclic convolution,
the pre- and post-additions in the bilinear algorithm correspond to
the pre-addition and post-addition modules in Fig.~\ref{fig:Av},
respectively, and the multiplications in the bilinear algorithm
correspond to the $k\times k$ matrix vector product (MVP) modules,
which compute the product between a $k\times k$ binary matrix and a
$k$-dimensional vector and can be achieved by simply additions.

The padding and reordering modules in Fig.~\ref{fig:Av} do not require
any logic, and the pre-addition and post-addition have a much smaller
complexity than the $k\times k$ MVP modules as shown in
Sec.~\ref{sec:complexity}. Therefore, the $k\times k$ MVP modules are
the primary source of the additive complexity of computing
$\mathbf{Av}$. To achieve high throughput, we can implement those
$k\times k$ modules in parallel. Furthermore, the CSE algorithm can
reduce the additive complexity of each module. Since $k$ is rather
small compared with $m$, the CSE algorithm is more effective in
simplifying a $k\times k$ MVP than an $m\times m$ one. However, as
each module corresponds to a different $k\times k$ matrix, the CSE
reduction results are different, and so are the addition networks for
each $k\times k$ MVP module. Therefore, those modules must be
implemented separately, and we cannot save any chip area by
implementing the circuitry in a serial or partly parallel fashion.

\begin{figure}[tb]
  \centering
  \includegraphics[width=0.9\columnwidth]{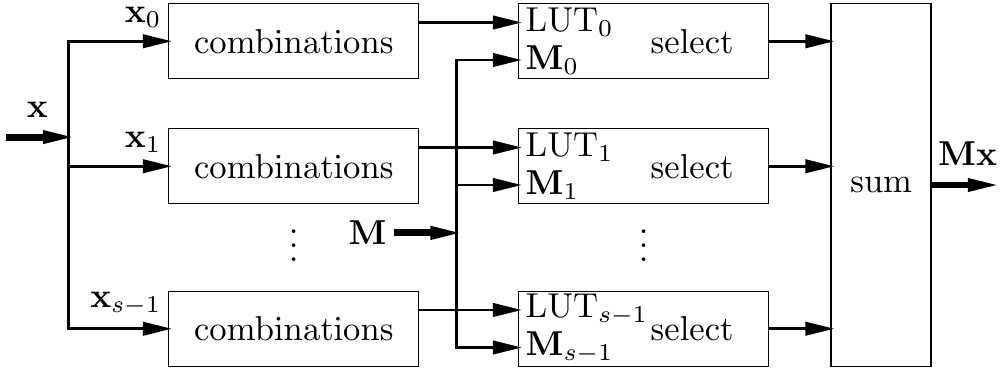}
  \caption{Our circuitry for $k\times k$ MVP modules using the
    Four-Russian algorithm. LUT stands for ``look-up table'' and
    $s=\lceil\log_2 k\rceil$.}
  \label{fig:FRA}
  \vspace{-4mm}
\end{figure}

To save area and power, the Four-Russian algorithm \cite{Aho1974} can
be used to implement the $k\times k$ modules in these cases, using the
architecture shown in Fig.~\ref{fig:FRA}. According the algorithm, the
circuitry has three stages to compute a $k\times k$ MVP, denoted as
$\mathbf{Mx}$. The first stage splits $\mathbf{x}$ into
$s=\lceil\log_2 k\rceil$ sub-vectors $\mathbf{x}_i$'s, and computes
all the binary combinations of elements in each $\mathbf{x}_i$. The
second stage partitions the matrix $\mathbf{M}$ into $1\times s$
sub-matrices $\mathbf{M}_i$'s accordingly, and computes
$\mathbf{M}_i\mathbf{x}_i$ by look-up tables generated from the first
stage. Finally the third stage sums up all
$\mathbf{M}_i\mathbf{x}_i$'s. Since the first and last stages in
Fig.~\ref{fig:FRA} are independent from $\mathbf{x}$ and $\mathbf{M}$,
they can be reused in a serial or partly parallel implementation to
save chip area. The second stage depends on $\mathbf{M}$, but it still
have a regular structure that is favorable in hardware
implementation. No memory or registers are needed in the fully
parallel implementation, and buffers used to hold the intermediate
results are required in the serial and partly parallel implementation.



In addition to its low complexity, the modular structures of the
architectures in Fig.~\ref{fig:Av} and Fig.~\ref{fig:FRA} are suitable
for hardware implementations. First, it is easy to apply architectural
techniques such as pipelining to these architectures for better clock
rate and throughput. Second, since the $k\times k$ MVP modules account
for the majority of the complexity, the modular structure provides
various tradeoff options between area, power, and throughput via
reusing the $k\times k$ MVP modules in Fig.~3 as well as the
combination and select modules in Fig.~4.


\bibliographystyle{IEEEtran}
\bibliography{IEEEabrv,FFT}

\end{document}